\documentclass[12pt,draftclsnofoot,journal,onecolumn]{IEEEtran}

\usepackage{setspace}
\usepackage{amsmath,url}
\usepackage{epsfig,amssymb,amsbsy,verbatim,array,enumerate}
\usepackage{pstricks,psfrag,theorem,cite,paralist,subfigure,empheq}
\usepackage{bm}
\usepackage{graphicx}
\usepackage{cite}
\usepackage{url}
\usepackage{color}

\let\intern=\iftrue

\def\figref#1{Fig.\,\ref{#1}}%
\def\E{\mathbb{E}}
\def\P{\mathbb{P}}
\def\R{\mathbb{R}}

\def\ie{{\em i.e.}}

\def\sinc{\operatorname{sinc}}
\def\dd{\mathrm{d}}

\def\SIR{\mathsf{SIR}}

\def\ps{p_{\rm s}}
\def\Ps{P_{\rm s}}

\def\eqa{\stackrel{{\rm (a)}}{=}}
\def\eqb{\stackrel{{\rm (b)}}{=}}
\def\eqc{\stackrel{{\rm (c)}}{=}}
\def\eqd{\stackrel{{\rm (d)}}{=}}

\newtheorem{theorem}{Theorem}
\newtheorem{lemma}{Lemma}
\newtheorem{corollary}{Corollary}
\newtheorem{definition}{Definition}

\newtheorem{remark}{Remark}

\newlength{\figwidth}


\usepackage{xcolor,tikz}

\let\arxiv\iftrue

\begin{document}
%
\title{SIR Meta Distribution of $K$-Tier Downlink Heterogeneous Cellular Networks with\\ Cell Range Expansion}
\author{Yuanjie~Wang, 
		Martin~Haenggi,~\textit{Fellow,~IEEE},
		and Zhenhui~Tan,~\textit{Member,~IEEE}
		\thanks{Yuanjie~Wang and Zhenhui~Tan are with the State Key Laboratory of Railway Traffic Control and Safety, Beijing Jiaotong University, Beijing, 100044, China (e-mail: wang.yuanjie@outlook.com). Martin~Haenggi is with the Dept. of Electrical Engineering, University of Notre Dame, IN, 46556, USA (e-mail: mhaenggi@nd.edu). The work was supported by National Natural Science Foundation of China (61471030), National Science and Technology Major Project of China (2015ZX03001027-003) and the US National Science Foundation (grant CCF 1525904). 
			
		Manuscript date February 28, 2018.}
		}
	
\maketitle
\doublespacing

\begin{abstract}
Heterogeneous cellular networks (HCNs) constitute a necessary step in the evolution of cellular networks. In this paper, we apply the signal-to-interference ratio (SIR) meta distribution framework for a refined SIR performance analysis of HCNs, focusing on $K$-tier heterogeneous cellular networks based on the homogeneous independent Poisson point process (HIP) model, with range expansion bias (offloading bias) in each tier. Expressions for the $b$-th moment of the conditional success probability for both the entire network and each tier are derived, based on which the exact meta distributions and the beta approximations are evaluated and compared. Key performance metrics including the mean success probability, the variance of the conditional success probability, the mean local delay and the asymptotic SIR gains of each tier are obtained. The results show that the biases are detrimental to the overall mean success probability of the whole network and that the $b$-th moment curve (versus the SIR threshold) of the conditional success probability of each tier can be excellently approximated by the horizontal shifted versions of the first moment curve of the single-tier PPP network. We also provide lower bounds for the region of the active probabilities of the base stations to keep the mean local delay of each tier finite.   
\end{abstract}

\begin{IEEEkeywords}
	Stochastic geometry, Poisson point process, heterogeneous cellular network, SIR, coverage, meta distribution, offloading. 	
\end{IEEEkeywords}

\section{Introduction}
\subsection{Motivation}
Heterogeneous cellular networks (HCNs), consisting of various types of base stations such as macro, pico and femto, are a necessary step in the evolution of cellular networks to meet the explosive demand in mobile data traffic growth and various emerging applications \cite{GhoshHCN12}. For seamless coverage, it is essential to understand the signal-to-interference ratio (SIR) distribution, especially at high deployment densities, which makes the network interference-limited. In the literature, the mathematical analysis for the SIR distribution in conventional single-tier cellular network and HCNs mainly relies on the application of Poisson point process (PPP) theory in stochastic geometry \cite{Haenggi18book, Andrews11, Dhillon12, Mukherjee12, nigam2014, Singh13, Jo12, Keeler13, Blaszczyszyn15, XZhang14}, which has been shown to be a powerful tool in recent years.  

However, the conventional SIR analysis for the HCNs is restricted to the mean success probability $\ps(\theta)\triangleq\P(\SIR>\theta)$, defined as the complementary cumulative distribution function (CCDF) of the SIR evaluated at the typical link. Such a performance metric is merely a macroscopic quantity by averaging the conditional success probability (link reliability) $\Ps(\theta)\triangleq \P(\SIR>\theta\mid\Phi)$ over the underlying point process $\Phi$, hence it provides no information about the difference between links. In contrast, the network operators' concerns for the real deployment of HCNs are questions such as ``How are the link reliabilities distributed among users in different tiers and/or in the whole network?'', or ``How will the offloading affect the SIR performance of different tiers?'', or ``What is the reliability level that the `5\% user'\footnote[1]{The ``5\% user'' refers to the user whose performance ranks at the 5th-percentile.} can achieve in each tier?''   

To obtain such fine-grained information on the SIR performance, the meta distribution concept was introduced in \cite{MHmeta}, which characterizes the distribution of the conditional success probabilities of the individual links given the point process. The lack of study of the meta distribution for HCNs with offloading biasing among different tiers motivates our study in this paper. We shall see that the meta distribution of SIR is a framework that facilitates the analysis for a series of performance metrics including the variances of the link reliability, the mean local delay and the asymptotic gains for HCNs.
      
\subsection{Related Work}
For the SIR-related analysis based on stochastic geometry in HCNs, the most commonly used model is the homogeneous independent Poisson (HIP) model, where BSs of each tier follow a homogeneous independent Poisson point process \cite[Def.~2]{Haenggi14wcl}. \cite{Jo12} utilized the HIP model with the (biased) nearest-BS association and considered offloading between different tiers, where offloading was implemented by biasing the transmit power of different tiers. \cite{Singh13} studied an extended heterogeneous network scenario where multiple radio access technologies (RATs) including cellular and Wi-Fi coexist, with each RAT consisting of multiple tiers and modeled by the HIP model, and the biasing association is also considered. The distribution of the SINR at the typical user was derived and applied to the analysis of rate coverage. In \cite{nigam2014}, coordinated multipoint joint transmission (CoMP) in HCN was analyzed and it was shown, as a special case (namely no-CoMP), that the result for a single tier in \cite{Andrews11} also holds for arbitrary tiers.

Instead of the (biased) nearest-BS association adopted in the above-mentioned works, there is also the line of work using the maximum instantaneous SINR association, such as \cite{Dhillon12, Mukherjee12, Keeler13, Blaszczyszyn15, XZhang14}. \cite{Dhillon12} studied the coverage (success) probability and the average rate of the HIP model for the SINR thresholds greater than $0$ dB under both open and closed access. \cite{Mukherjee12} utilized the HIP model and determined the coverage probability from the joint CCDF of the SINR at the typical user with the SINR thresholds extended to all regime. \cite{Keeler13} and \cite{Blaszczyszyn15} also extended the SINR threshold to less than 0 dB and established the exact results for the maximum instantaneous SINR association rule with arbitrary shadowing in HCNs by the $K$-coverage probability. As for the fading model, it should be noted that different from \cite{Dhillon12}, where only Rayleigh fading is considered, it has been shown that the same result applies to arbitrary fading in \cite{XZhang14}.

As for modeling the HCNs with more general point processes, \cite{Deng15} proposed two models for the two-tier HCN with the inter-tier independence modeled by combining the PPP and the Poisson hole process, and the intra-tier independence taken into account by combining the PPP and Matern cluster process respectively, yielding more accurate results for the outage probability and the area spectral efficiency. In \cite{WeiASAPPP16}, for HCNs consisting of general point processes as each tier with unbiased association, the authors studied the SIR distribution by using the shifted versions of the PPP SIR distributions as approximations.

Most of these above-discussed works related to SIR analysis in HCNs only analyze the mean success probability without delving into the SIR performance at the individual link level. To overcome this limitation, we need to develop the meta distribution framework for the HCNs.  

The meta distribution has been applied to different scenarios since it was formally formulated in \cite{MHmeta}, where the analysis of single-tier Poisson bipolar networks with ALOHA channel access and the downlink of Poisson cellular networks laid the foundation of the concept. It was applied to study D2D communication underlaid with the downlink of Poisson cellular networks \cite{2017mHaenggiTcom}, uplink and downlink Poisson cellular networks with fractional power control \cite{YWmetaPower}, D2D communications with interference cancellation \cite{YWmetaIC}, millimeter-wave D2D networks \cite{Deng17}, the spatial outage capacity \cite{SOCmeta}, and downlink coordinated multi-point transmission/reception (CoMP) in cellular networks \cite{CuiCoMP17}. These studies revealed some interesting new insights that are of significance to the deployment of real networks.  

\subsection{Contributions}
In this paper, we develop an SIR meta distribution analysis framework for the HIP downlink model under Rayleigh fading. We show that this framework enables a comprehensive understanding of a series of key performance metrics and network design problems. Specifically, 
\begin{itemize}
	\item We derive exact analytical expressions of the $b$-th moment of the conditional success probability for both the overall typical user and the typical user in each tier under Rayleigh fading. 
	\item We show that the beta distribution is an excellent approximation for the exact meta distribution of both the entire network and each tier.
	\item We reveal that both the $b$-th moment and the variance of the conditional success probability for each tier can be efficiently approximated by horizontally shifting the mean success probability curve of the single-tier PPP according to the asymptotic SIR gains, whose expressions are given explicitly.  
	\item We rigorously study the effects of the offloading biases on both the entire network and each tier in terms of the first moment and variance of the conditional success probability. 
	\item We extend the model to include random base station activity by ALOHA and derive analytical expressions of the $b$-th moment of the conditional success probability for both the overall typical user and the typical user in each tier.
	\item We derive lower bounds of the region of ALOHA probabilities so that the mean local delay remains finite under the effect of random base station activity.   
\end{itemize} 

\subsection{Organization}
The rest of the paper is organized as follows: Section~\ref{sec:SystemModel} introduces the system model and the concept of the SIR meta distribution in HCNs. Section~\ref{sec:MainResults} develops the general framework for the analysis of HCNs using the meta distribution, wherein we derive the exact analytical expressions of the $b$-th moment of the conditional success probability, both for the entire network and for each individual tier, and discuss various key performance metrics and some network design problems related to offloading. Section~\ref{sec:BSactivity} extends the SIR meta distribution to the analysis of random base station activity. Section~\ref{sec:Conclusion} concludes the paper.
 
\section{System Model}\label{sec:SystemModel}
\subsection{SIR Model}
We consider a general $K$-tier heterogeneous cellular network model, where BSs of each tier follow a homogeneous independent Poisson point process $\Phi_i$ with intensity $\lambda_i$. This is the so-called homogeneous independent Poisson (HIP) model \cite[Def.~2]{Haenggi14wcl}. For the BSs of the $i$-th tier, the transmit power is $P_i$, and the range expansion bias is $B_i$. For BS ${\rm x}\in \Psi=\bigcup\limits_{i\in[K]}\Phi_i$, $\iota({\rm x})\in[K]$, denotes its tier number and $[K]=\{1,2,...K\}$. We assume the standard power-law path loss model with exponent $\alpha>2$, and define $\delta = 2/\alpha$. The downlink association rule is the biased nearest-BS association, \ie, for the typical user at the origin $o$, its serving BS $\nu(o)$ is drawn from all BSs according to 
\begin{equation}\label{eq:nu_o}
\nu(o)= \mathop{\arg\max}\limits_{{\rm x}\in\Psi}\{P_{\iota({\rm x})}B_{\iota({\rm x})}\|{\rm x}\|^{-\alpha}\},
\end{equation}  
where $\iota({\rm x})$ is the tier index of BS $\rm x$.

The power fading coefficient associated with BS ${\rm x}\in\Psi$ is denoted by $h_{\rm x}$, which is exponentially distributed with $\E(h_{\rm x})=1$ (Rayleigh fading). $R_j$ is the distance from the typical user to the nearest BS in $\Phi_j$. First we focus on the fully loaded case on a certain resource block (RB), \ie, all BSs are always active on the RB in consideration.

Letting ${\rm x_0} = \nu(o)$, for the typical user at the origin, the received signal-to-interference ratio (SIR) is given by
\begin{equation}\label{eq:SIR_typ}
\SIR_o = \frac{P_{\iota({\rm x_0})}h_{\rm x_0}\|\rm x_0\|^{-\alpha}}{\sum\limits_{{\rm x}\in\Psi\setminus\{\rm x_0\}}P_{\iota({\rm x})}h_{\rm x}\|{\rm x}\|^{-\alpha}}.
\end{equation}

\subsection{Meta Distribution for HCNs}
The SIR meta distribution for single-tier cellular networks is the two-parameter function defined as \cite{MHmeta}
\begin{equation}
\bar F(\theta, t) \triangleq \bar F_{\Ps}(t) = \P(\Ps(\theta)>t),\quad \theta\in\R^+,\:t\in [0,1],
\label{eq:meta}
\end{equation}
which is the CCDF of the conditional success probability (link reliability) $\Ps$. The $b$-th moment of the meta distribution is denoted by $M_b(\theta)\triangleq \E(\Ps(\theta)^b)$.

We consider two types of SIR meta distributions, one is for the overall network (\ie, the overall typical user) and the other is specific to the $i$-th tier, obtained by conditioning on the typical user connecting to that tier. In the following, we use the label $(i)$ for the quantities related to the $i$-th tier meta distributions.

\section{SIR Meta Distribution Framework}\label{sec:MainResults}
In this section we derive the general analytical expression for the $b$-th moment of the meta distribution in the HIP model with biasing.   

\subsection{Moments of the Conditional Success Probability}
First, we state a lemma about the conditional and average access probabilities for the typical user connecting to the given $i$-th tier, which is a slight reformulation of \cite[Lemma 1]{Jo12}. Hence the proof is omitted.
\allowdisplaybreaks
\begin{lemma}[Access probability]
	\label{lem:AccessProb}
	Defining $\iota(\rm x_0)\triangleq\iota(\nu(o))$, the conditional access probability for the typical user connecting to the $i$-th tier given $R_i$ is
	\begin{equation}\label{eq:condi_acc_Prob}
	\P(\iota({\rm x_0})=i \mid R_i) = \prod_{j\neq i} e^{-\lambda_j \pi ({\hat P_{ij}} {\hat B_{ij}})^{\delta} R_i^2},
	\end{equation}	
	and the access probability that the typical user is associated with the $i$-th tier is 
	\begin{equation}\label{eq:AccessProb}
	p_{\rm a}^{(i)}\triangleq\P(\iota({\rm x_0})=i) = \frac{1}{\sum\limits_{j\in[K]} \hat{\lambda}_{ij}(\hat{P}_{ij}\hat{B}_{ij})^\delta} 
	\end{equation}
	where $\hat\lambda_{ij} = \lambda_j/\lambda_i$, $\hat P_{ij} = P_j/P_i$ and $\hat B_{ij} = B_j/B_i$.
\end{lemma}  
%
Next we present the first main result on the moments of the conditional success probability.
\begin{theorem}[Moments for the $K$-tier HCNs]
	\label{thm:Mb_RE}
	For the overall typical user in the $K$-tier HIP model with range expansion, the $b$-th moment of the conditional success probability is given by
	\begin{equation}
	\label{eq:Mb-typ-Bias}
	M_b = \sum\limits_{i} \frac{1}{\sum\limits_{j}\hat\lambda_{ij} (\hat P_{ij} \hat B_{ij})^\delta~ _2F_1(b,-\delta; 1-\delta; -\theta\hat B_{ij}^{-1})}.
	\end{equation}
	where $i,~j\in[K]$, $\hat\lambda_{ij} = \lambda_j/\lambda_i$, $\hat P_{ij} = P_j/P_i$ and $\hat B_{ij} = B_j/B_i$.
\end{theorem}
\begin{IEEEproof}
	See Appendix A.
\end{IEEEproof}

\begin{corollary}[Moments without range expansion]
	\label{thm:Mb-WO-RE}
	For the overall typical user, the $b$-th moment $M_b$ with no range expansion in any tier, \ie, $B_i=1$ for $i\in[K]$, is given by
	\begin{equation}\label{eq:Mb_typ}
	M_b=\frac{1}{_2F_1(b,-\delta; 1-\delta; -\theta)} ,\quad b\in \mathbb{C}.
	\end{equation}
\end{corollary}
\begin{IEEEproof}
This can be easily obtained by setting $B_i=1$ for $i\in[K]$ in \eqref{eq:Mb-typ-Bias}.
\end{IEEEproof} 
\begin{remark}
The $b$-th moment of the meta distribution of the overall typical user in a HIP-based $K$-tier downlink HCN without range expansion in any tier is the same as that in a single-tier network \cite[Thm.~2]{MHmeta}. Hence the meta distribution is the same. This shows that the multitier architecture does not improve the performance of the 5\% user (or, more generally, the fairness between the users).
\end{remark}

\begin{corollary}[Moments for the typical user in the $\bm i$-th tier]
	\label{thm:Mb-ith-tier}
	Conditioned on the typical user connecting to the $i$-th tier, the $b$-th moment of the meta distribution is given by
	\begin{equation}\label{eq:Mb-ith-tier}
	M_{b \mid (i)} = \frac{\sum_{j} \hat{\lambda}_{ij}(\hat{P}_{ij}\hat{B}_{ij})^\delta}{\sum_j \hat{\lambda}_{ij}(\hat{P}_{ij}\hat{B}_{ij})^\delta ~_2F_1(b,-\delta; 1-\delta; -\theta \hat B_{ij}^{-1})}.
	\end{equation}
	where $\hat\lambda_{ij} = \lambda_j/\lambda_i$, $\hat P_{ij} = P_j/P_i$ and $\hat B_{ij} = B_j/B_i$.
\end{corollary}   
\begin{IEEEproof}
This follows directly from the proof of Thm.~\ref{thm:Mb_RE}.
\end{IEEEproof}

\begin{corollary}[Mean local delay]
	For the typical user in the $i$-th tier, the mean local delay is given by
	\begin{equation}\label{eq:MLD_ith}
	M_{-1\mid (i)} = \frac{(1-\delta)\sum_j \hat \lambda_{ij} (\hat P_{ij} \hat B_{ij})^\delta}{\sum_j \hat \lambda_{ij} (\hat P_{ij} \hat B_{ij})^\delta (1-\delta-\delta\theta \hat B_{ij}^{-1})},
	\end{equation}
\end{corollary}
\begin{IEEEproof}
	The mean local delay is the -1-st moment of the conditional success probability in Cor.~\ref{thm:Mb-ith-tier}. Using the identity $_2F_1(-1,b;c;z) \equiv 1-\frac{bz}{c}$, \eqref{eq:MLD_ith} is obtained.
\end{IEEEproof}

The mean local delay $M_{-1\mid (i)}$ has a phase transition at $\theta_{{\rm c}\mid(i)}$ as given in \eqref{eq:theta_c} when it is seen as a function of the SIR threshold with the other parameters fixed, which means the mean local delay is finite for $\theta < \theta_{{\rm c}\mid(i)}$ and is infinite for $\theta \geq \theta_{{\rm c}\mid(i)}$.
\begin{equation}\label{eq:theta_c}
\theta_{{\rm c}\mid(i)}=\frac{(1-\delta)\sum_j \hat \lambda_{ij} (\hat P_{ij} \hat B_{ij})^\delta}{\delta \sum_j \hat \lambda_{ij} (\hat P_{ij})^\delta (\hat B_{ij})^{\delta-1}},
\end{equation}

\subsection{Approximations of the Meta Distribution}
According to the Gil-Pelaez theorem\cite{gil1951note}, for a general variable $X>0$ with characteristic function $\varphi_X(t)\triangleq\E e^{jtX}$, $j\triangleq\sqrt{-1}$, $t\in\R$, the CCDF of $X$ is given by 
\begin{equation}\label{eq:Gil}
\bar F_X(x)=\frac12+\frac1\pi\int_0^\infty \frac{\Im(e^{-jt\log x}\varphi_X(jt))}{t}\dd t, 
\end{equation}
where $\Im(z)$ denotes the imaginary parts of $z\in\mathbb{C}$.

Letting $X\triangleq\Ps(\theta)$ (or $X\triangleq P_{s\mid(i)}(\theta)$), we have $\varphi_X(t)=M_{jt}$ (or $\varphi_X(t)=M_{jt\mid(i)}$), setting $b=jt$ in \eqref{thm:Mb_RE} (or \eqref{eq:Mb-ith-tier}). Hence, the meta distribution of the conditional success probability for the whole network (and the specific $i$-th tier) can be calculated.

Calculation of the exact meta distribution via the Gil-Pelaez theorem usually involves many calculations of imaginary moments, which prohibits direct insights into the meta distributions and its applications in mapping to other performance metrics like the ergodic data rate \cite{Deng17}, etc. An efficient approximation of the meta distribution is obtained by using the beta distribution through matching their first and second moments, which has been verified in \cite{MHmeta, Deng17, 2017mHaenggiTcom, YWmetaPower, YWmetaIC} for various network scenarios. 
  
\subsection{Asymptotic SIR Gains}

As shown in \cite{Haenggi14wcl,AGuo15ADG,Ganti16}, the CCDFs $\bar F_{\SIR}(\theta)$ of the SIR at the typical user in different general single-tier nearest-associated networks resemble merely horizontally shifted versions in the SIR threshold $\theta$ (in dB) of each other, as long as they have the same diversity gain. The horizontal gap (or the ``SIR gain'') relative to a reference network model at the target success probability $p_{\rm t}$ is given by 
\begin{equation}\label{eq:HorizonGap}
	G_{\rm p}(p_{\rm t}) \triangleq \frac{\bar F_{\SIR}^{-1}(p_{\rm t})}{\bar F_{\SIR_{\rm ref}}^{-1}(p_{\rm t})},
\end{equation}
where $\bar F_{\SIR}^{-1}$ is the inverse function of $\bar F_{\SIR}(\theta)$. 

Usually it is more convenient to write $G_{\rm p}(p_{\rm t})$ as a function of $\theta$ by $G(\theta)=\theta'/\theta$, 
where $\theta'$ is given by $\bar F_{\SIR}(\theta')=\bar F_{\SIR_{\rm ref}}(\theta)=p_{\rm t}$.

The asymptotic SIR gain at the high-reliability regime is defined by
\begin{equation}\label{eq:HorizonGap3}
G_0 \triangleq \lim_{\theta\to 0}G(\theta).
\end{equation}

Similarly, the asymptotic SIR gain at the low-reliability regime is defined as
\begin{equation}\label{eq:HorizonGapInf}
G_\infty \triangleq \lim_{\theta\to \infty}G(\theta).
\end{equation}

Usually, the most sensible reference network model is the homogeneous PPP. If $G_0$ (or $G_\infty$) exists, then a rather convenient way to estimate $p_{\rm s}(\theta)$ of the network in focus is by using $G_0$ (or $G_\infty$) as the scaling factor $G$ for $\theta$, \ie,
\begin{equation}\label{eq:HorizonGap4}
p_{\rm s}(\theta) \approx p_{\rm s, PPP}({\theta/G}).
\end{equation}

$G(\theta)$ in dB quantifies the horizontal gap between $p_{\rm s}(\theta)$ and $p_{\rm s, PPP}(\theta)$ for $\theta$ in dB.

Next, we extend the above-mentioned SIR asymptotic gain in single-tier networks to HCNs based on the HIP model.
\begin{definition}[Asymptotic SIR gains in HCNs]
	For the HCN model in this paper, the asymptotic SIR gains of the $b$-th moment of the conditional success probability for each tier, at both the high-reliability and low-reliability regimes, with the standard success probability of the single-tier PPP as the reference, are, respectively, given by 
\begin{equation}\label{eq:AsymptGainMb}
G_{0,b}^{(i)} = \lim_{\theta\to 0}\frac{M_{b\mid(i)}^{-1}(p_{\rm s, PPP}(\theta))}{\theta},
\end{equation}
and 	
\begin{equation}\label{eq:AsymptGainMbInf}
G_{\infty,b}^{(i)} = \lim_{\theta\to \infty}\frac{M_{b\mid(i)}^{-1}(p_{\rm s, PPP}(\theta))}{\theta}.
\end{equation}	
where $M_{b\mid(i)}^{-1}$ is the inverse function of $M_{b\mid(i)}$ and $p_{\rm s, PPP}(\theta) = M_1$ in \eqref{eq:Mb_typ}.
\end{definition}

We will show that, remarkably, the horizontal shift is applicable to each tier in the HCN. Before deriving the asymptotic gains, we first state a lemma about the asymptotics of the hypergeometric function $_2F_1$. 
\begin{lemma}\label{lem:Asymp2F1}
	For $b\in\mathbb{C}$,	
	\begin{equation}\label{eq:Asymp2F1_0}
	_2F_1(b,\delta,1-\delta;-z) \sim 1+bz\frac{\delta}{1-\delta}, ~~z\to 0,
	\end{equation}
	and
	\begin{equation}\label{eq:Asymp2F1_Inf}
	_2F_1(b,\delta,1-\delta;-z) \sim z^\delta T(b), ~~z\to \infty,
	\end{equation}
	where $T(b)=\int_0^\infty(1-(1+r^{-\frac{1}{\delta}})^{-b})\dd r$.
\end{lemma}

\begin{IEEEproof}
	By Taylor expansion, at $z=0$,
	\begin{equation}\label{eq:Taylor}
	\frac{1}{(1+z)^b} \sim 1-bz,
	\end{equation}	
	hence
	\begin{align}\label{eq:2F1_asymp}
	_2F_1(b,\delta,1-\delta;-z) &= 1+ \int_1^\infty \Big(1-\frac{1}{(1+z s^{-1/\delta})^b} \Big)\dd s \nonumber \\
	&\sim 1+ \int_1^\infty \Big(1-(1- bz s^{-1/\delta}) \Big)\dd s \nonumber \\
	&= 1+bz\frac{\delta}{1-\delta}.
	\end{align}
	\allowdisplaybreaks
	When $z\to \infty$, we have
	\begin{align}
	_2F_1(b,-\delta; 1-\delta; -z) &= 1 + 2\int_0^1 \Big(1-\frac{1}{(1+z r^\alpha)^b}\Big) r^{-3} \dd r \nonumber \\
	&= 1+ z^\delta 2\int_0^{z^{\frac{1}{\alpha}}} \Big(1-\frac{1}{(1+r^\alpha)^b}\Big) r^{-3} \dd r \nonumber \\
	&\sim z^\delta 2\int_0^{\infty} \Big(1-\frac{1}{(1+r^\alpha)^b}\Big) r^{-3} \dd r \nonumber \\
	&\sim z^\delta \int_0^{\infty} \Big(1-\frac{1}{(1+r^{-\frac{1}{\delta}})^b}\Big) \dd r.
	\end{align}	
	where the first step is according to \cite[eq. (23)]{MHmeta}; the second step follows variable substitution $z^{\frac{1}{\alpha}}r\to r$; the third step follows since $z\to \infty$ and the last step follows variable substitution $r\to r^{-\frac{1}{2}}$.
\end{IEEEproof}

\begin{corollary}[Asymptotic SIR gains relative to PPP]
	\label{cor:AsymptGain}
	Conditioned on the typical user connecting to the $i$-th tier, the asymptotic SIR gains of the $b$-th moment of the meta distribution relative to $M_1$ of the single-tier homogeneous PPP are given by
	\begin{equation}\label{eq:AsymptGain}
	G_{0,b}^{(i)} = \frac{\sum_{j} \hat{\lambda}_{ij}\hat{P}_{ij}^\delta \hat{B}_{ij}^\delta}{b\sum_j \hat{\lambda}_{ij}\hat{P}_{ij}^\delta \hat{B}_{ij}^{\delta-1}},
	\end{equation}
	and
	\begin{equation}\label{eq:AsymptGainInf}
	G_{\infty,b}^{(i)} = \Big(\frac{T(1)}{T(b)}\frac{\sum_{j} \hat{\lambda}_{ij}\hat{P}_{ij}^\delta \hat{B}_{ij}^\delta}{\sum_j \hat{\lambda}_{ij}\hat{P}_{ij}^\delta}\Big)^{\frac{1}{\delta}},
\end{equation}	
	where $b\in\mathbb{C}$, $\hat\lambda_{ij} = \lambda_j/\lambda_i$, $\hat P_{ij} = P_j/P_i$ and $\hat B_{ij} = B_j/B_i$.
\end{corollary}   
\begin{IEEEproof}	
	To determine $G_{0,b}^{(i)}$, we need to evaluate the limit of $M_{b\mid(i)}(\theta)$ at $\theta\to 0$. 
	Applying \eqref{eq:Asymp2F1_0} in \eqref{eq:Mb-ith-tier},
	\begin{align}\label{eq:Mb_ith_asymp}
	M_{b\mid(i)}(\theta) &\sim \frac{\sum_{j} \hat{\lambda}_{ij}\hat{P}_{ij}^\delta \hat{B}_{ij}^\delta}{\sum_{j} \hat{\lambda}_{ij}\hat{P}_{ij}^\delta \hat{B}_{ij}^\delta \Big(1+b\theta\hat{B}_{ij}^{-1} \frac{\delta}{1-\delta} \Big)} \nonumber\\
	&=\frac{1}{1+ \frac{\delta}{1-\delta}\frac{\theta}{G_{0,b}^{(i)}}} \nonumber\\
	&\sim 1- \frac{\delta}{1-\delta}\frac{\theta}{G_{0,b}^{(i)}}.
	\end{align}	
	
	Since for the PPP, 
	\begin{equation}\label{eq:M1_PPP_asymp}
	M_{1,{\rm PPP}}(\theta) = \frac{1}{_2F_1(1,\delta,1-\delta;-\theta)} \sim 1-\frac{\theta\delta}{1-\delta},
	\end{equation}	
	it is clear that $G_{0,b}^{(i)}$ is exactly the asymptotic gain for $\theta\to 0$.
	
	To determine $G_{\infty,b}^{(i)}$, applying \eqref{eq:Asymp2F1_Inf} in \eqref{eq:Mb-ith-tier}, we have
	\begin{equation}\label{eq:Mb_ith_asympinf}
	M_{b\mid(i)}(\theta) \sim \bigg(\frac{\sum_{j} \hat{\lambda}_{ij}\hat{P}_{ij}^\delta}{\sum_{j} \hat{\lambda}_{ij}\hat{P}_{ij}^\delta \hat{B}_{ij}^\delta}T(b)\theta^\delta\bigg)^{-1}. 
	\end{equation}
	$G_{\infty,b}^{(i)}$ is then obtained by comparing \eqref{eq:Mb_ith_asympinf} and \eqref{eq:M1_PPP_asymp}.		
\end{IEEEproof}

\begin{remark}
	In Cor.~\ref{cor:AsymptGain}, the reference model in use is the first moment of the conditional success probability of the single-tier PPP. Another efficient way is to use the $b$-th moment the conditional success probability of the single-tier PPP as the reference model, then the variable $b$ in \eqref{eq:AsymptGain} and \eqref{eq:AsymptGainInf} vanishes and the two asymptotic gains become constants. From this it is easy to be inferred that the variances $V^{(i)}(\theta)$ of each tier are also shifted versions of each other, as shown in \figref{fig:HorizonShiftVar_B2_10} in Sec.~\ref{sec:Apps}. 
\end{remark} 

\subsection{Base Station Activity}\label{sec:BSactivity} 
In this section, we model the random activities of interfering base stations in each tier by the ALOHA model, \ie, the interfering BSs of tier $i$ are active only with probability $p_i$. The activities of different base stations are independent. We first derive the general $b$-th moment for the typical user of each individual tier and the whole network, and then the lower bound of the activity probabilities to keep the mean local delay finite.  

\begin{theorem}
	\label{thm:Mb-ith-tier_BSActivity}
	Given that the typical user connects to the $i$-th tier with the serving BS always being active, and the interfering BSs in tier $j\in[K]$ are active independently with probability $p_j$, the $b$-th moment of the meta distribution can be expressed as
	\begin{equation}\label{eq:Mb-ith-tier-BSActivity}
	M_{b \mid (i)}(\bm p) = \frac{\sum_{j} \hat{\lambda}_{ij}(\hat{P}_{ij}\hat{B}_{ij})^\delta}{\sum\limits_j \hat{\lambda}_{ij}(\hat{P}_{ij}\hat{B}_{ij})^\delta \big(1-\sum\limits_{k=1}^{\infty}\binom bk (-p_j \theta \hat B_{ij}^{-1})^k \frac{\delta}{k-\delta} ~_2F_1(k,k-\delta; k-\delta+1; -\theta \hat B_{ij}^{-1})\big)}.
	\end{equation}
	where $\bm p=(p_1,p_2,...p_K)$, $\hat\lambda_{ij} = \lambda_j/\lambda_i$, $\hat P_{ij} = P_j/P_i$, and $\hat B_{ij} = B_j/B_i$.
\end{theorem}
\begin{IEEEproof}
	See Appendix B.
\end{IEEEproof}
	As expected, letting $K=1$, \eqref{eq:Mb-ith-tier-BSActivity} retrieves the single-tier result in \cite[Thm.~3]{MHmeta}; also, letting $K=2$ and the two tiers share the same parameters, the result of each tier is also the same as the-single tier result.
 
\begin{theorem}
	\label{thm:Mb-typ_BSActivity}
	For the overall typical (active) user with the interfering BSs in tier $j\in[K]$ are active independently with probability $p_j$, the $b$-th moment of the meta distribution can be expressed as
	\begin{equation}\label{eq:Mb-typ-BSActivity}
	M_{b}(\bm p) = \sum\limits_i \frac{1}{\sum\limits_j \hat{\lambda}_{ij}(\hat{P}_{ij}\hat{B}_{ij})^\delta \big(1-\sum\limits_{k=1}^{\infty}\binom bk (-p_j \theta \hat B_{ij}^{-1})^k \frac{\delta}{k-\delta} ~_2F_1(k,k-\delta; k-\delta+1; -\theta \hat B_{ij}^{-1})\big)}.
	\end{equation}
	where $\bm p=(p_1,p_2,...p_K)$, $\hat\lambda_{ij} = \lambda_j/\lambda_i$, $\hat P_{ij} = P_j/P_i$, and $\hat B_{ij} = B_j/B_i$.
\end{theorem} 

From \eqref{eq:Mb-ith-tier-BSActivity}, the mean local delay of the typical user connecting to the $i$-th tier is given by 
\begin{equation}\label{eq:MLDp}
M_{-1\mid(i)}(\bm p) = \frac{1}{D_i(\bm p)},~~{\bm p}\in \mathcal{S}_i,
\end{equation}where 
\begin{align}\label{eq:Di}
D_i(\bm p) &= 1-\frac{p_i\theta\delta}{1-\delta}\:_2F_1(1,1-\delta;2-\delta;-\theta(1-p_i)) \nonumber\\
&~~~~+\sum\limits_{j\neq i}\frac{\lambda_j}{\lambda_i}\Big(\frac{P_jB_j}{P_iB_i}\Big)^\delta\Big(1-\frac{p_j\theta\delta}{1-\delta}\:_2F_1(1,1-\delta;2-\delta;-\theta B_iB_j^{-1}(1-p_j))\Big),
\end{align}
and $\mathcal{S}_i$ is the region for $\bm p$ in which the mean local delay is finite for the $i$-th tier, defined by
\begin{equation}\label{eq:FiniteRegion}
	\mathcal{S}_i\triangleq\{(p_1,p_2,...p_K)\in[0,1]^K:D_i(\bm p)>0\} .
\end{equation} 
The boundary of the region for the finite mean local delay for the $i$-th tier is then defined as
\begin{equation}\label{eq:FiniteBoundary}
 \partial\mathcal{S}_i\triangleq\{(p_1,p_2,...p_K)\in[0,1]^K:D_i(\bm p)=0\} .
\end{equation}

The region of all tiers is then given by the intersection
\begin{equation}\label{eq:FiniteRegionAll}
\mathcal{S}\triangleq\bigcap\limits_{i\in[K]}\mathcal{S}_i.
\end{equation} 

A simple but reasonable inference from \eqref{eq:MLDp} and \eqref{eq:Di} is that for small $\bm p$, the mean local delay is finite since the interference is low and most of the users in each tier have a high conditional success probability, as $\bm p$ grows higher, the interference gets severe, and with $\bm p$ increasing to some critical threshold, $D_i(\bm p)$ will go to zero, resulting in the infinite mean local delay.

It is hard to exactly characterize $\mathcal{S}_i$, next we provide a lower bound $\partial\mathcal{\check S}_i$ of $\mathcal{S}_i$ to shed light on the effect of the base station activity probabilities $\bm{p}$. By noticing that 
\begin{align}
_2F_1(1,1-\delta;2-\delta;-z) &\eqa (1+z)^{-1} \:_2F_1\Big(1,1;2-\delta;\frac{z}{1+z}\Big) \nonumber\\
&\eqb (1+z)^{-1} \sum\limits_{m=0}^{\infty} \frac{(1)_m (1)_m}{(2-\delta)_m}\frac{u^m}{m!} \nonumber\\
&= (1+z)^{-1} \sum\limits_{m=0}^{\infty} \frac{(1)_m}{(2-\delta)_m} u^m \nonumber\\
&< (1+z)^{-1}\Big(1+ \frac{1}{2-\delta}u +\frac{1}{2-\delta}u^2 + ... \Big) \nonumber\\
&= (1+z)^{-1}\Big(1+\frac{1}{2-\delta}\frac{u}{1-u}\Big) \nonumber\\
&= (1+z)^{-1}\Big(1+\frac{1}{2-\delta}z\Big),
\end{align}  
where (a) is by the Euler's transformation; (b) is by the series form of the Gaussian hypergeometric function $_2F_1$ and  $(q)_m\equiv\frac{\Gamma(q+m)}{\Gamma(q)}$ is the Pochhammer function (rising factorial). The boundary is given by
\begin{equation}\label{eq:FiniteInnerBound}
\partial\mathcal{\check S}_i=\{(p_1,p_2,...p_K)\in[0,1]^K:\check D_i(\bm p)=0\},
\end{equation}
where 
\begin{equation}\label{eq:checkDfun}
\check D_i(\bm p) =  \sum\limits_j\frac{\lambda_j}{\lambda_i}\bigg(\frac{P_jB_j}{P_iB_i}\Big)^\delta \Big(1-\frac{p_j\theta\delta}{1-\delta}\Big(1+\theta\frac{B_i}{B_j}(1-p_j)\Big)^{-1}\Big(1+\frac{\theta B_i(1-p_j)}{(2-\delta)B_j}\Big)\bigg).
\end{equation}
 
\section{Applications in Two-tier HCNs}\label{sec:Apps}
In this section, we apply the meta distribution framework developed in Sec.~\ref{sec:MainResults} to the two-tier HIP model and show the corresponding numerical results. Since the performances are affected only by the ratios between the densities, transmit powers and biases of the two tiers, we assume $P_1=\lambda_1=B_1=1$ without loss of generality.
\subsection{Moments}
Defining $f_b(x)\triangleq \:_2F_1(b,-\delta; 1-\delta; -x)$, we obtain the first moment and variance for each tier from Cor.~\ref{thm:Mb-ith-tier}, 
\begin{equation}\label{eq:M1tier1}
M_{1\mid (1)} = \frac{1+\lambda_2 (P_2 B_2)^\delta}{f_1(\theta)+\lambda_2 (P_2 B_2)^\delta~ f_1(\theta B_2^{-1})},
\end{equation}
\begin{equation}\label{eq:M1tier2}
M_{1\mid (2)} = \frac{1+\lambda_2^{-1} (P_2 B_2)^{-\delta}}{f_1(\theta)+\lambda_2^{-1} (P_2 B_2)^{-\delta}~ f_1(\theta B_2)},
\end{equation}
\begin{equation}\label{eq:V1}
V_{(1)} = \frac{1+\lambda_2 (P_2 B_2)^\delta}{f_2(\theta)+\lambda_2 (P_2 B_2)^\delta~ f_2(\theta B_2^{-1})}-\Big(\frac{1+\lambda_2 (P_2 B_2)^\delta}{f_1(\theta)+\lambda_2 (P_2 B_2)^\delta~ f_1(\theta B_2^{-1})}\Big)^2,
\end{equation}
\begin{equation}\label{eq:V2}
V_{(2)} = \frac{1+\lambda_2^{-1} (P_2 B_2)^{-\delta}}{f_2(\theta)+\lambda_2^{-1} (P_2 B_2)^{-\delta}~ f_2(\theta B_2)}-\Big(\frac{1+\lambda_2^{-1} (P_2 B_2)^{-\delta}}{f_1(\theta)+\lambda_2^{-1} (P_2 B_2)^{-\delta}~ f_1(\theta B_2)}\Big)^2.
\end{equation}
\begin{figure} [!t]
	\begin{center}
		\includegraphics[width=0.95\figwidth]{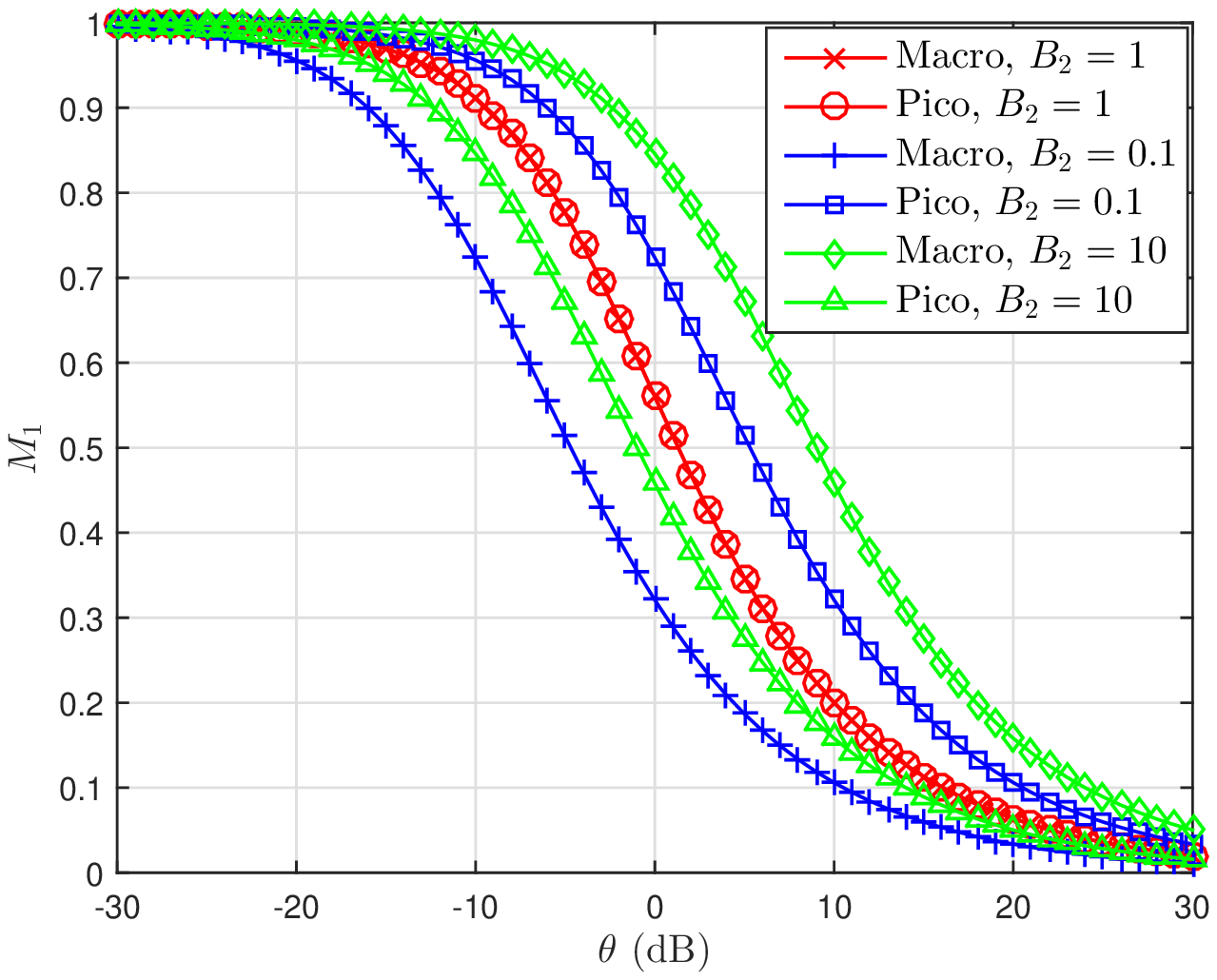}
		\caption{$M_1$ of the typical user in each tier versus $\theta$ with $\alpha=4$, $P_2=0.2$ and $\lambda_2=5$. In this case, for $B_2=1$, $p_{\rm a}^{(1)}=0.5$ and $p_{\rm a}^{(2)}=0.5$; for $B_2=0.1$, $p_{\rm a}^{(1)}=0.59$ and $p_{\rm a}^{(2)}=0.41$; for $B_2=10$, $p_{\rm a}^{(1)}=0.12$ and $p_{\rm a}^{(2)}=0.88$.}
		\label{fig:M1vsTheta}
	\end{center}
\end{figure} 
\begin{figure} [!t]
	\begin{center}
		\includegraphics[width=0.95\figwidth]{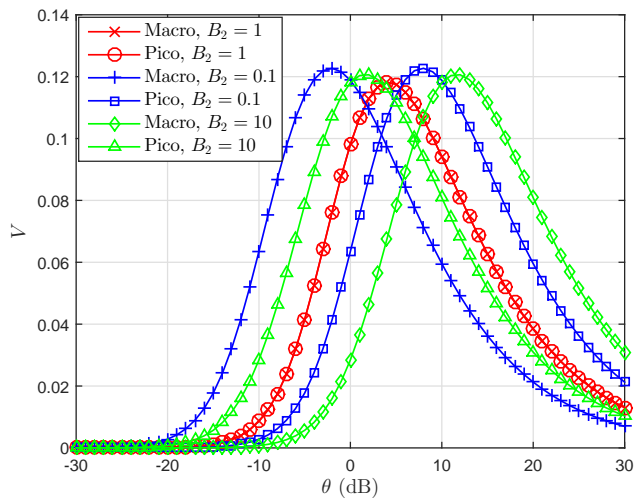}
		\caption{$V$ of the typical user in each tier versus $\theta$ with $\alpha=4$, $P_2=0.2$ and $\lambda_2=5$. In this case, for $B_2=1$, $p_{\rm a}^{(1)}=0.5$ and $p_{\rm a}^{(2)}=0.5$; for $B_2=0.1$, $p_{\rm a}^{(1)}=0.59$ and $p_{\rm a}^{(2)}=0.41$; for $B_2=10$, $p_{\rm a}^{(1)}=0.12$ and $p_{\rm a}^{(2)}=0.88$.}
		\label{fig:VarvsTheta}
	\end{center}
\end{figure}

\figref{fig:M1vsTheta} and \figref{fig:VarvsTheta} show $M_1$ and $V$ of each tier in a two-tier HCN. We can see that if there is no bias (\ie, $B_1=B_2=1$), the curves of $M_1$ and $V$ of both tiers coincide, which implies that the two tiers have the same SIR statistics regardless of their different densities and powers. However, the inequality in range expansion bias results in the separation between these two tiers in terms of $M_1$ and $V$. Specifically, since  biasing means offloading, we can draw the conclusion that offloading from one tier to the other will always benefit $M_1$ of the former, while harming the latter for any given $\theta$.  

\subsection{Beta Approximations} 
In \figref{fig:MetaDistr_theta0dB_B2_10dB}, we see that this approximation is also excellent for HCNs with biases.
\begin{figure} [!t]
	\begin{center}
		\includegraphics[width=0.95\figwidth]{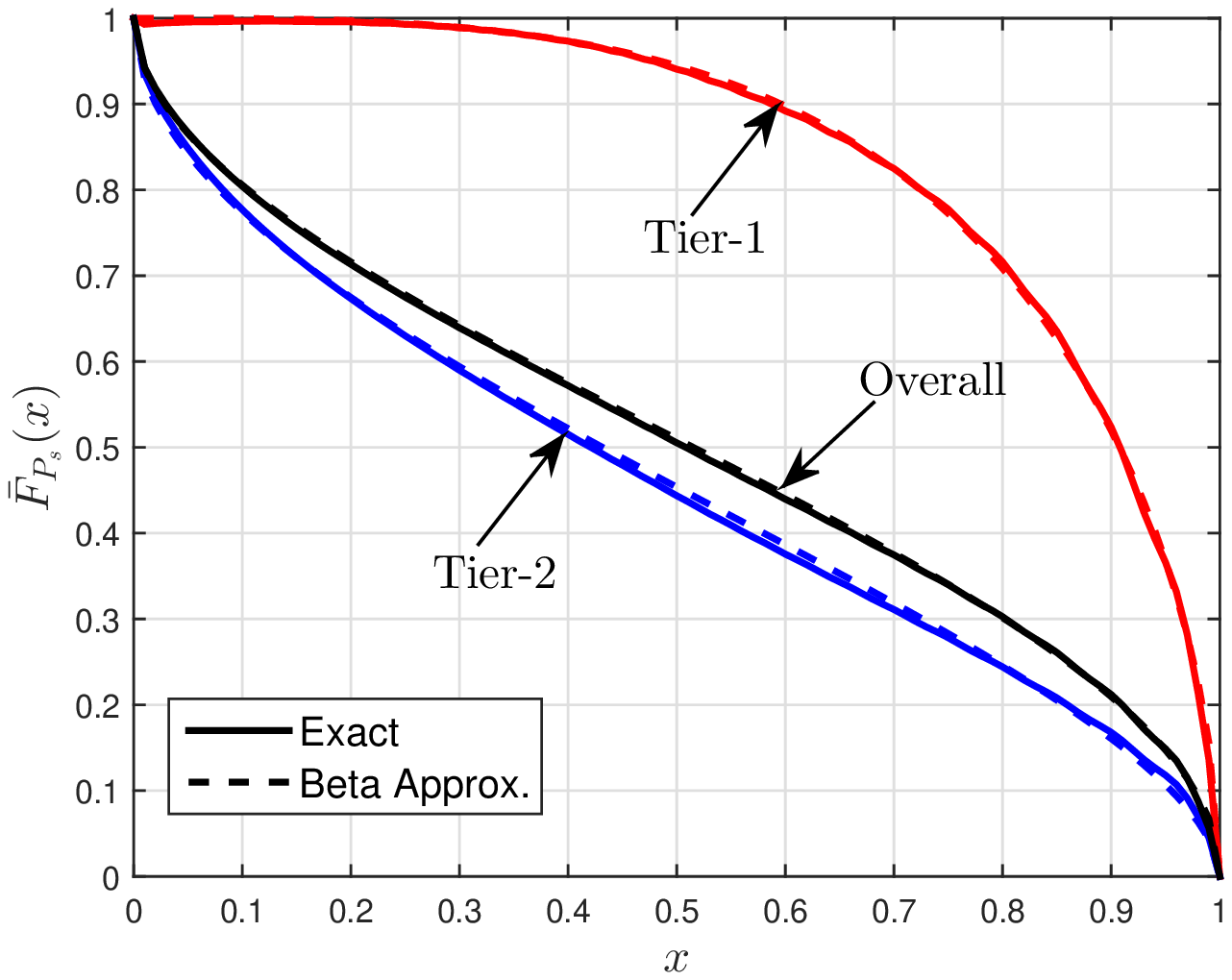}
		\caption{The exact meta distribution for the overall network and for each tier of a two-tier HCN with $\theta=0$ dB, $\alpha=4$, $\lambda_1=1$, $\lambda_2=5$, $P_1=1$, $P_2=0.2$, $B_1=1$ and $B_2=10$. The solid lines correspond to the exact results and the dashed lines are the beta approximations.}
		\label{fig:MetaDistr_theta0dB_B2_10dB}
	\end{center}
\end{figure} 

\subsection{Horizontal Shifting via Asymptotic SIR Gains} 
For the two-tier HCN example, the asymptotic SIR gain for $M_1$ of each tier is respectively given by $G_{0,1}^{(1)} = \frac{1+\lambda_2 P_2^\delta B_2^\delta}{1+\lambda_2 P_2^\delta B_2^{\delta-1}}$ and $G_{0,1}^{(2)} = \frac{1+\lambda_2^{-1} P_2^{-\delta} B_2^{-\delta}}{1+\lambda_2^{-1} P_2^{-\delta} B_2^{1-\delta}}$.
Numerically, for the case $B_2=10$ dB shown in \figref{fig:M1vsTheta} and \figref{fig:VarvsTheta}, $G_{0,1}^{(1)} = 6.75$ dB, $G_{0,1}^{(2)} = -3.25$ dB, $G_{0,2}^{(1)} = 3.74$ dB,  $G_{0,2}^{(2)} = -6.27$ dB, $G_{\infty,1}^{(1)} = 9.94$ dB, $G_{\infty,1}^{(2)} = -2.06$ dB, $G_{\infty,2}^{(1)} = 4.42$ dB and $G_{\infty,2}^{(2)} = -5.58$ dB. \figref{fig:HorizonShift2_B2_10} shows the comparison between the exact $b$-th moment curves and the shifted versions of the $M_1$ of a single-tier PPP as the reference model and \figref{fig:HorizonShiftVar_B2_10} shows the comparison between the exact variance curves and the shifted versions of the variance of a single-tier PPP as the reference model. We can see that the shifted versions by using the asymptotic gain are excellent approximations for the exact results. 
\begin{figure} [!t]
	\begin{center}
		\includegraphics[width=0.95\figwidth]{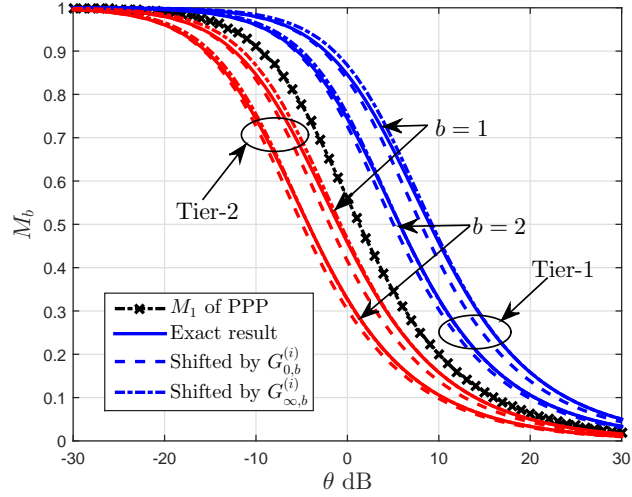}
		\caption{Illustration for the asymptotic gain of $M_b$ in each tier of a two-tier HCN relative to $M_1$ of a single-tier PPP. In this case, $\alpha=4$, $P_2=0.2$ and $\lambda_2=5$, $B_2=10$ dB. The solid lines correspond to the exact results and the dashed lines are the shifted versions of $M_1$ of the single-tier PPP by $G_{0,b}^{(i)}$ and $G_{\infty,b}^{(i)}$, $i=1,2$, respectively.}
		\label{fig:HorizonShift2_B2_10}
	\end{center}
\end{figure}
\begin{figure} [!t]
	\begin{center}
		\includegraphics[width=0.95\figwidth]{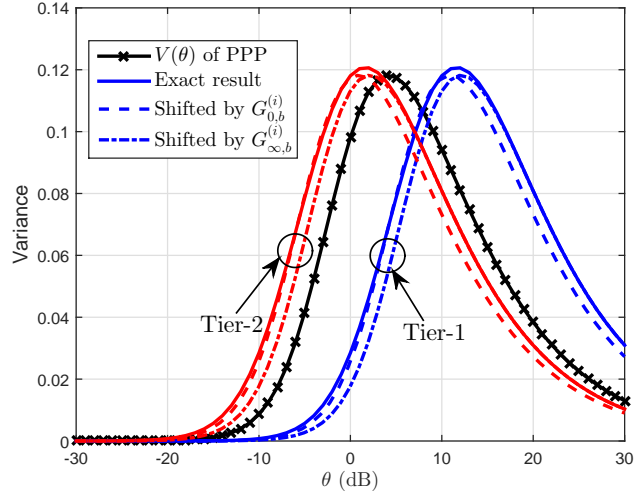}
		\caption{Illustration for the asymptotic gain of $V^{(i)}(\theta)$ of a two-tier HCN relative to the variance of a single-tier PPP. In this case, $\alpha=4$, $P_2=0.2$ and $\lambda_2=5$, $B_2=10$ dB. The solid lines correspond to the exact results and the dashed lines are the shifted versions of $V(\theta)$ of the single-tier PPP by the asymptotic gains $G_{0,b}^{(i)}$ and $G_{\infty,b}^{(i)}$, $i=1,2$, respectively.}
		\label{fig:HorizonShiftVar_B2_10}
	\end{center}
\end{figure}
\subsection{Effects of Biasing} 
In this section, we study the effects of range expansion biases on the coverage performance of each individual tier and the whole network. 
 
Sometimes, it is convenient and of significance to consider the asymptotic performance of the range expansion biases.
\begin{corollary}\label{cor:AsympB2}
	For $B_2\to \infty$, which means that tier 1 is closed-access, we have
	\begin{itemize}[\rm (a)]
		\item[\rm (a)] $\displaystyle M_{1\mid (1)}\sim 1,~ M_{1\mid (2)}\sim \frac{\lambda_2 P_2^\delta \sinc\delta}{F(\delta,\theta)\lambda_2 P_2^\delta \sinc\delta + \theta^\delta}$; further, for $\displaystyle \theta\to\infty,~ M_{1\mid (2)}\sim \frac{\lambda_2 P_2^\delta \sinc\delta}{\theta^\delta(1+\lambda_2 P_2^\delta)}$;
		\vspace{2pt}
		\item[\rm (b)] $\displaystyle V_{(1)}\to 0,~ V_{(2)}\sim \frac{1}{F(\delta,\theta)+\frac{3-2\delta}{(2-\delta)\lambda_2 P_2^\delta \sinc\delta}\theta^\delta}-\Big(\frac{\lambda_2 P_2^\delta \sinc\delta}{F(\delta,\theta)\lambda_2 P_2^\delta \sinc\delta + \theta^\delta}\Big)^2$; further, for $\displaystyle \theta\to\infty,~ V_{(2)}\sim \frac{\lambda_2 P_2^\delta \sinc\delta}{\theta^\delta (1+\delta)(1+\lambda_2 P_2^\delta)} - \frac{\lambda_2^2 P_2^{2\delta} \sinc^2\delta}{\theta^{2\delta} (1+\lambda_2 P_2^\delta)^2}$,
	\end{itemize}
	where $F(\delta,\theta)=\:_2F_1(1,-\delta; 1-\delta; -\theta)$. 
\end{corollary}
\begin{IEEEproof}
	These results are easily obtained by using \eqref{eq:Asymp2F1_Inf} in Lem.~\ref{lem:Asymp2F1} by noting that $T(1)=\frac{1}{\sinc\delta}$, $T(2)=\frac{1+\delta}{\sinc\delta}$ and the identity $_2F_1(a, b; c; 0)\equiv 1$.
\end{IEEEproof}

\begin{corollary}\label{cor:EffectBiasM1}
$B_2>1 \Leftrightarrow M_{1\mid (1)}>M_{1\mid (2)}$. 
\end{corollary}
\begin{IEEEproof}
   Since $f_1(x)$ is monotonically increasing, we have $f_1(\theta B_2)>f_1(\theta)>f_1(\theta B_2^{-1})$ for $B_2>1$. Then from \eqref{eq:M1tier1} and \eqref{eq:M1tier2} we have $M_{1\mid (1)}>\frac{1}{1+f_1(\theta)}$ while $M_{1\mid (2)}<\frac{1}{1+f_1(\theta)}$.
\end{IEEEproof}

In words, offloading from one tier to the other will harm the average success probability of the latter tier. 

As for the overall typical user, according to Thm.~\ref{thm:Mb_RE}, its first moment and variance of the conditional success probability are, respectively, given by 
\begin{equation}
M_1(B_2) = \frac{1}{f_1(\theta)+\lambda_2 (P_2 B_2)^\delta f_1(\theta B_2^{-1})} + \frac{1}{f_1(\theta)+\lambda_2^{-1} (P_2 B_2)^{-\delta}f_1(\theta B_2)},
\end{equation}
\begin{align}
V(B_2) &= \frac{1}{f_2(\theta)+\lambda_2 (P_2 B_2)^\delta f_2(\theta B_2^{-1})} + \frac{1}{f_2(\theta)+\lambda_2^{-1} (P_2 B_2)^{-\delta} f_2(\theta B_2)} \nonumber \\
&~ -\Bigg(\frac{1}{f_1(\theta)+\lambda_2 (P_2 B_2)^\delta f_1(\theta B_2^{-1})} + \frac{1}{f_1(\theta)+\lambda_2^{-1} (P_2 B_2)^{-\delta} f_1(\theta B_2)} \Bigg)^2.
\end{align} 

We can prove that $\frac{\partial M_1}{\partial B_2}\Big |_{B_2=1} = 0$, $\frac{\partial V}{\partial B_2}\Big |_{B_2=1} = 0$, which means $B_2=1$ is an extreme point. Also, $\frac{\partial^2 M_1}{\partial B_2^2}\Big |_{B_2=1} \leq 0$, hence $B_2=1$ is the maximal point of $M_1$ (see that shown in \figref{fig:M1B2} ). For the second derivative of $V$ at $B_2=1$, it is not easy to judge its sign across different values of $B_2$ since it is related to the value of $\theta$. But we can observe this from the analytical curves shown in \figref{fig:VarV2} that $B_2 = 1$ is the local minimum.
\begin{figure} [!t]
	\begin{center}
		\includegraphics[width=\figwidth]{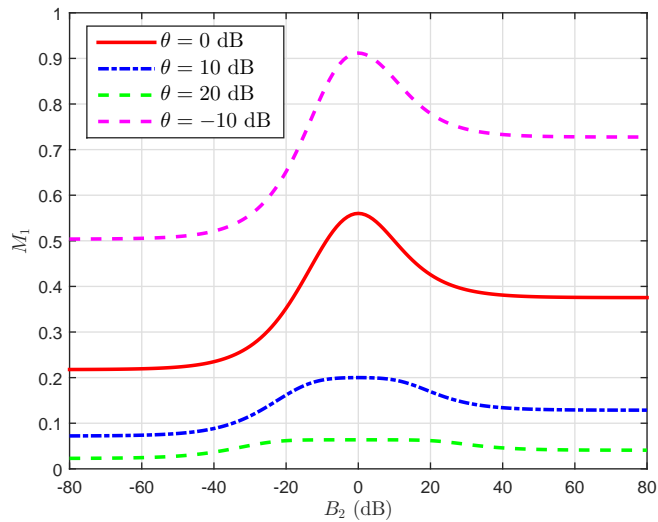}
		\caption{Analytical results for $M_1$ of the typical user of the entire network versus $B_2$ with $\alpha=4$, $P_2=0.2$ and $\lambda_2=4$.}
		\label{fig:M1B2}
	\end{center}
\end{figure} 
\begin{figure} [!ht]
	\begin{center}
		\includegraphics[width=\figwidth]{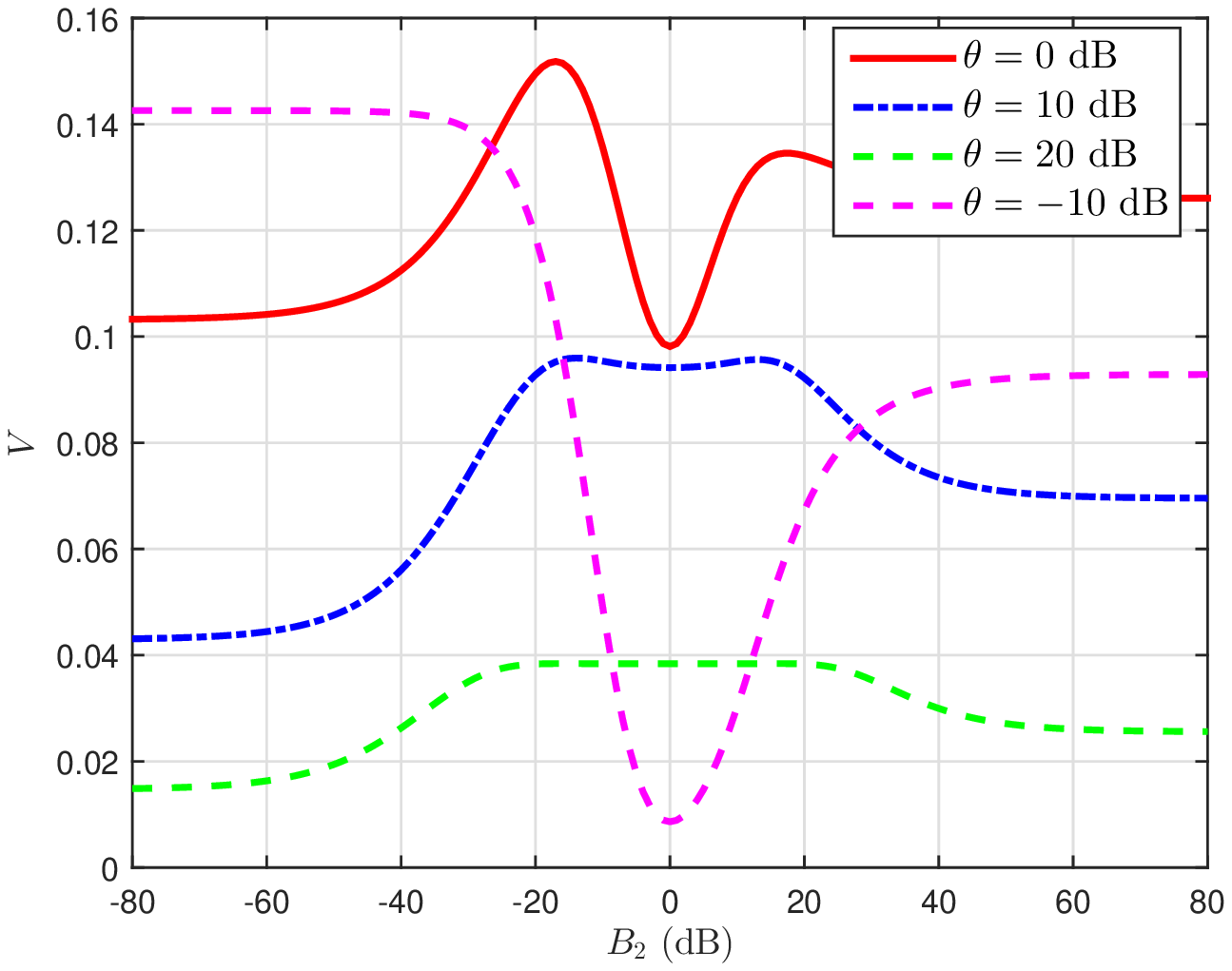}
		\caption{Analytical results for $V$ of the typical user of the entire network versus $B_2$ with $\alpha=4$, $P_2=0.2$ and $\lambda_2=4$.}
		\label{fig:VarV2}
	\end{center}
\end{figure}

\begin{figure} [!t]
	\begin{center}
		\includegraphics[width=0.95\figwidth]{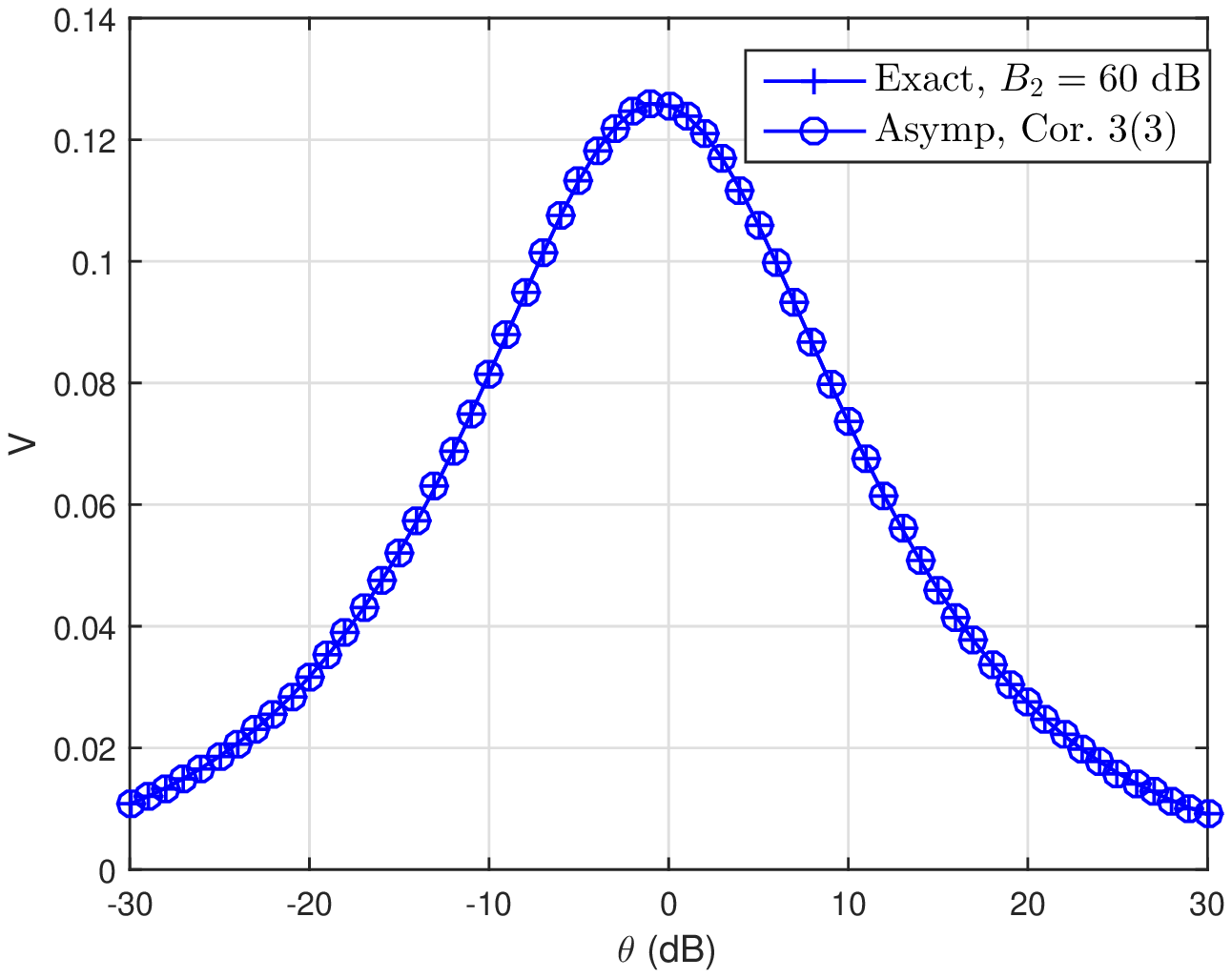}
		\caption{Asymptotic $V$ of the typical user in the pico tier versus $\theta$ with $\alpha=4$, $P_2=0.2$ and $\lambda_2=5$.}
		\label{fig:Asym_V2}
	\end{center}
\end{figure}

\begin{figure*}[!t]
	\begin{minipage}{0.45\linewidth}
		\centerline{\includegraphics[width=0.95\figwidth]{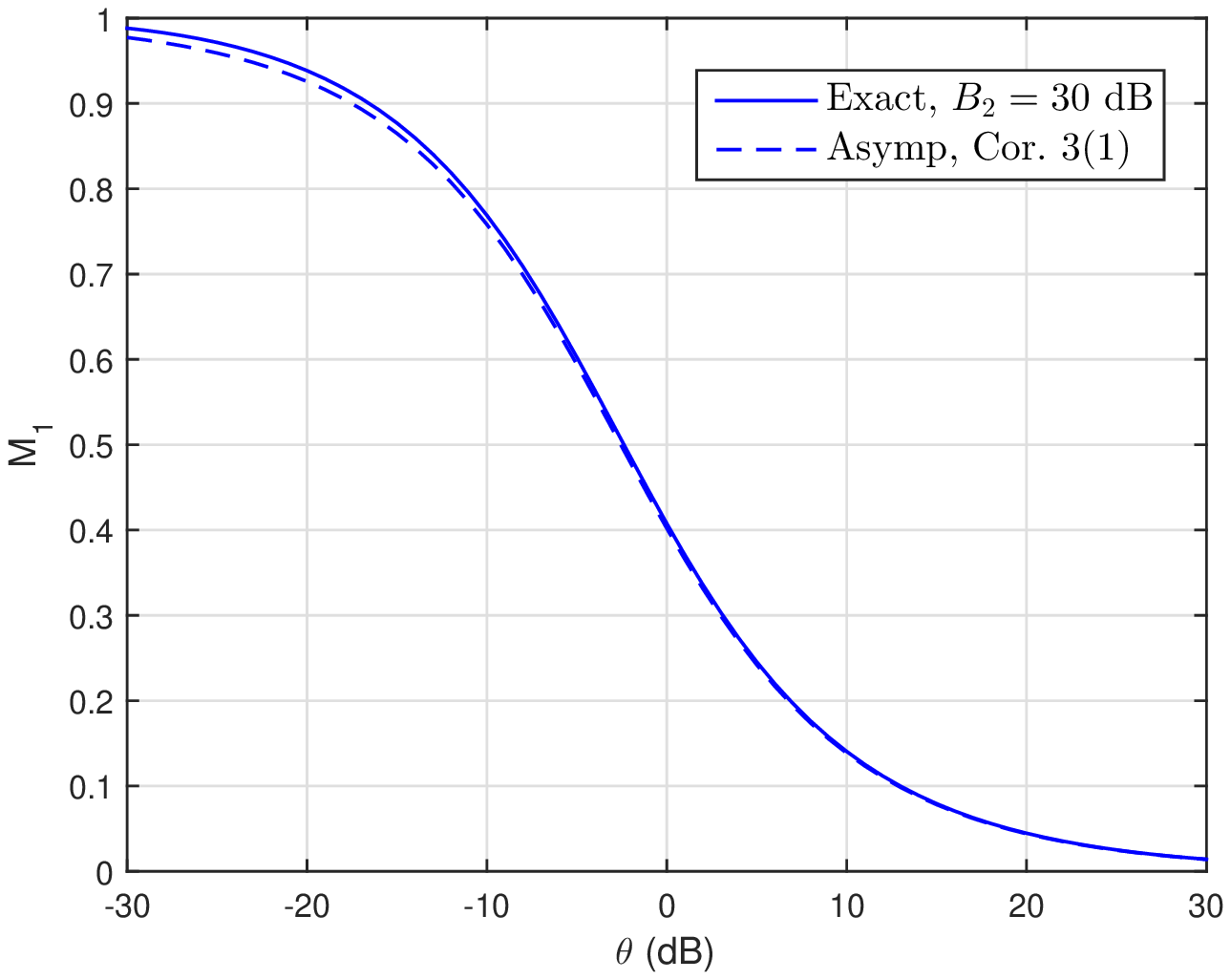}}
		\caption{Asymptotic $M_1$ of the typical user in the pico tier versus $\theta$ with $\alpha=4$, $P_2=0.2$ and $\lambda_2=5$.}
		\label{fig:Asym_M1_B2_30dB}
	\end{minipage}
	\hfill
	\begin{minipage}{0.45\linewidth}
		\centerline{\includegraphics[width=0.95\figwidth]{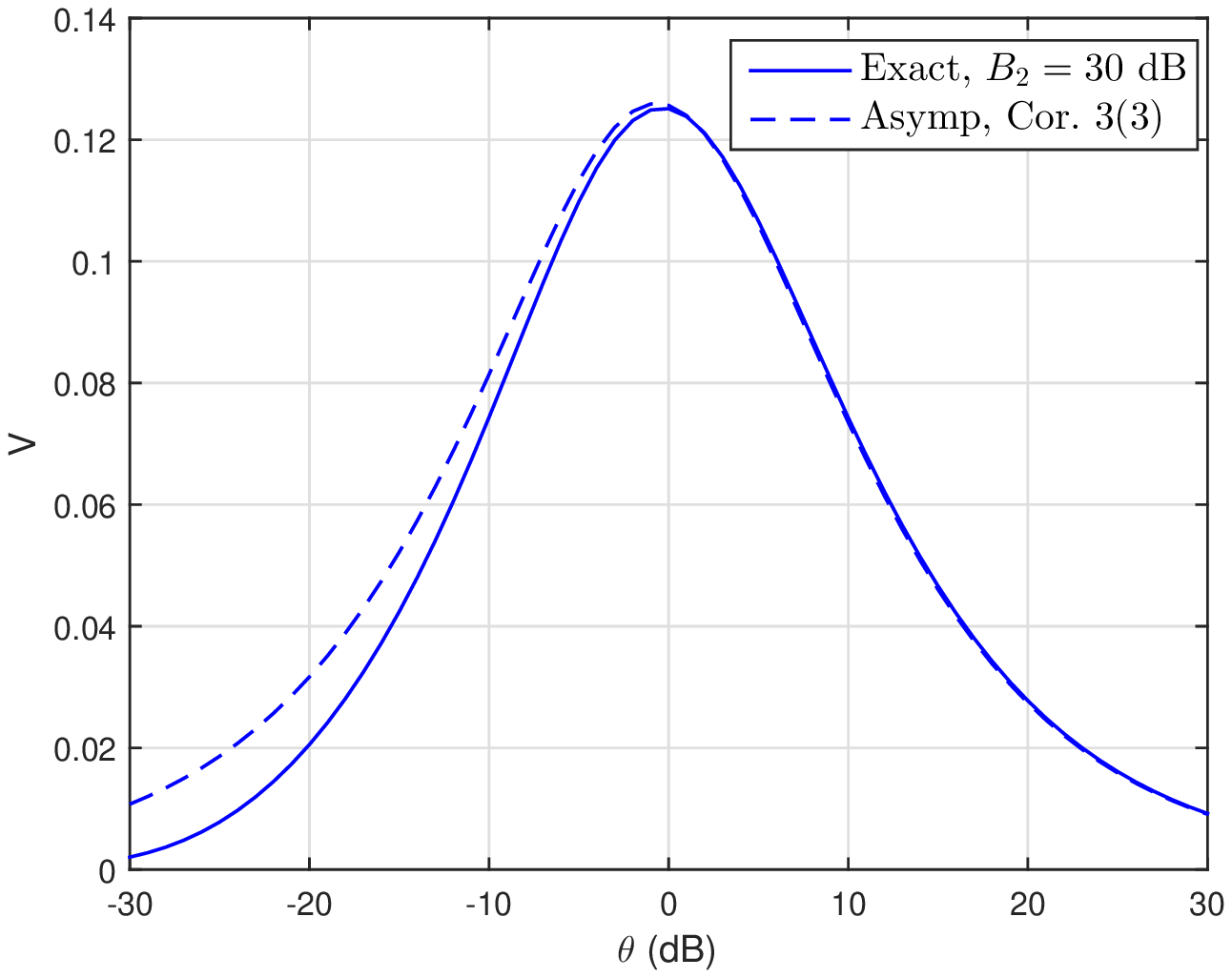}}
		\caption{Asymptotic $V$ of the typical user in the pico tier versus $\theta$ with $\alpha=4$, $P_2=0.2$ and $\lambda_2=5$.}
		\label{fig:Asym_V_B2_30dB}
	\end{minipage}
\end{figure*}

Based on the above analysis, for the $M_1$ of the overall users in a general $K$-tier HCN, we have the following corollary.
\begin{corollary}
	For the $K$-tier HIP model, setting all bias terms $B_i$ to the same value (i.e., no biasing) maximizes $M_1(\theta)$ of the overall typical user for all $\theta>0$.
\end{corollary} 
\begin{IEEEproof}
	For an arbitrary realization of the point process $\Psi$, determine the local-average SIR, which is equal to \eqref{eq:SIR_typ} but without the fading coefficients (see \cite[Eqn. (11)]{George17}) for all users for $B_i=1, i\in[K]$ (no biasing). This is by definition the best local-average SIR that each user can achieve. Consequently, if for any tier $i$, $B_i\neq 1$, there will be some users whose local-average SIR will decrease since they are no longer associated with the strongest-on-average BS. This implies that $M_1$ decreases.
\end{IEEEproof}

\begin{remark}
	For a general $K$-tier HCN with range expansion bias $B_i$ in the $i$-th tier, it is not easy to determine whether $B_i>1$ is harmful to the coverage performance in terms of $M_1$ for the $i$-th tier than the case with $B_i=1$. Since what play the decisive role are the ratios between $B_i$ and the bias values of the other tiers, which, in essence, reflect the offloading relationship among different tiers. In particular, $B_i/B_j < 1$ means offloading from the $i$-th tier to the $j$-th tier and vice versa. Hence, for a two-tier case, if some of the users in the first tier are offloaded to the second tier, then the latter definitely suffers a loss in $M_1$; however, for a three-tier case, if some of the users in the first tier are offloaded to the second tier, but some users belong to the second tier are also offloaded to the third tier, then for the second tier, its $M_1$ may improve.  
\end{remark}

\subsection{Lower Bounds of the Mean Local Delay with Random BS Activity}
Specifically, for a two-tier HIP model, we have the following corollary.
\begin{corollary}\label{cor:2tierProbBound}
	For a two-tier HCN, given all the other parameters,
	\begin{enumerate}[(1)]
		\item if $B_1=B_2$, then $\mathcal{S}=\mathcal{S}_1 = \mathcal{S}_2$;
		\item if $B_i>B_j$, then $\mathcal{S}=\mathcal{S}_j$, $i,j\in\{1,2\}$;
		\item if $\theta < \frac{1-\delta}{\delta}$, then $\mathcal{S}=\mathcal{S}_1 = \mathcal{S}_2=[0,1]^2$.
	\end{enumerate}
\end{corollary}
\begin{proof}
	For a two-tier HCN, we have
	\begin{align}\label{eq:D_1st}
	D_1(p_1,p_2) &= \underbrace{1-\frac{p_1\theta\delta}{1-\delta} \:_2F_1(1,1-\delta;2-\delta;-\theta(1-p_1))}_{A_1} \nonumber\\
	&~~+ \frac{\lambda_2}{\lambda_1}\Big(\frac{P_2B_2}{P_1B_1}\Big)^\delta \bigg(\underbrace{1-\frac{p_2\theta\delta}{1-\delta} \:_2F_1\Big(1,1-\delta;2-\delta;-\theta(1-p_2)\frac{B_1}{B_2}\Big)}_{G_2}\bigg),
	\end{align}
	\begin{align}\label{eq:D_2nd}
	D_2(p_1,p_2) &= \underbrace{1-\frac{p_2\theta\delta}{1-\delta} \:_2F_1(1,1-\delta;2-\delta;-\theta(1-p_2))}_{G_1} \nonumber\\
	&~~+ \frac{\lambda_1}{\lambda_2}\Big(\frac{P_1B_1}{P_2B_2}\Big)^\delta \bigg(\underbrace{1-\frac{p_1\theta\delta}{1-\delta} \:_2F_1\Big(1,1-\delta;2-\delta;-\theta(1-p_1)\frac{B_2}{B_1}\Big)}_{A_2}\bigg).
	\end{align}
	\begin{enumerate}[(1)]
		\item For $B_1=B_2$, let $g(x) = 1-\frac{x\theta\delta}{1-\delta} \:_2F_1(1,1-\delta;2-\delta;-\theta(1-x))$, $c=\frac{\lambda_2}{\lambda_1}\big(\frac{P_2}{P_1}\big)^\delta$, then $D_1(p_1,p_2) = g(p_1) + cg(p_2)$, $D_2(p_1,p_2) = \frac{D_1(p_1,p_2)}{c}$, since $c>0$, it is obvious that $D_1(p_1,p_2)$ and $D_2(p_1,p_2)$ always get negative at the same $(p_1,p_2)$. Hence $\mathcal{S}_1$ and $\mathcal{S}_2$ share the same boundary and thus $\mathcal{S}_1=\mathcal{S}_2$. 
		\item Without loss of generality, we assume $B_2>B_1$. Let $d=\frac{B_2}{B_1}>1$, then $D_1(p_1,p_2) = A_1+cd^\delta G_2$, $D_2(p_1,p_2) = \frac{A_2+cd^\delta G_1}{cd^\delta}$. Since $_2F_1(1,1-\delta;2-\delta;-z)$ is a monotonically decreasing function of $z$ for $z\geq0$, which is easy to be proved by its first-order derivative, for given $p_1, p_2$, we have $A_1<A_2$, $G_1>G_2$, hence as $p_1$ and (or) $p_2$ increase, $D_1(p_1,p_2)$ will decrease to zero first, resulting in $\mathcal{S}_1\subset\mathcal{S}_2$.  
		\item Let $p_1=p_2=1$, then $\check D_i(1,1) = \big(1+\sum_{j\neq i} \frac{\lambda_j}{\lambda_i} \big(\frac{P_j B_j}{P_i B_i}\big)^\delta\big) \big(1-\theta\frac{\delta}{1-\delta}\big)$, $\check D_i(1,1)>0$ requires $\theta<\frac{1-\delta}{\delta}$. 
	\end{enumerate}\vspace{-0.8cm}
\end{proof}

\begin{figure} [!t]
	\begin{center}
		\includegraphics[width=0.95\figwidth]{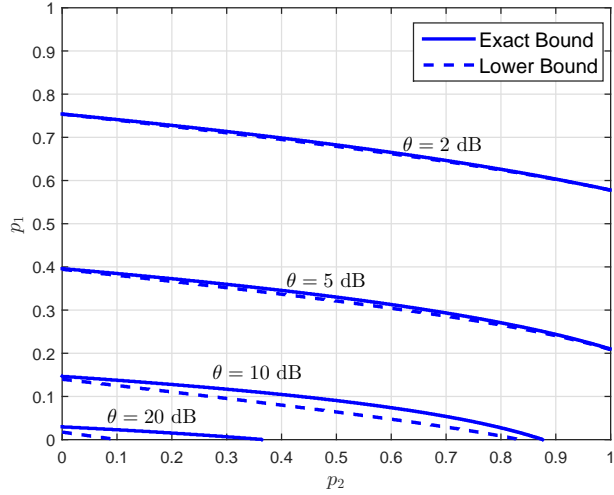}
		\caption{The exact boundary $\partial\mathcal{S}_1$ and its lower bound $\partial\mathcal{\check S}_1$ of a two-tier HCN with $\alpha=4$, $\lambda_2/\lambda_1=25$, $P_1/P_2=200$ and $B_2/B_1=10$. In this case, $\mathcal{S}=\mathcal{S}_1$.}
		\label{fig:InnerBound}
	\end{center}
\end{figure}  
\vspace{0.5cm}
In \figref{fig:InnerBound}, the exact boundary $\partial\mathcal{S}_1$ and its lower bound $\partial\mathcal{\check S}_1$ of a two-tier HCN are shown. As we see, the lower bound becomes tighter as $\theta$ decreases. In this case, according to Cor.~\ref{cor:2tierProbBound}(2), $\mathcal{S}=\mathcal{S}_1$. We also observe that as $\theta$ decreases, $\mathcal{S}$ grows towards $[0,1]^2.$ 

\section{Conclusions}\label{sec:Conclusion}  
In this paper, we developed the SIR meta distribution framework for the analysis of HIP-based $K$-tier HCNs with offloading biases and Rayleigh fading and performed a systematic study for a series of key performance metrics, revealing fine-grained information on the per-user performance. We first derived the $b$-th moment of the conditional success probability for both the entire network and each single tier. Based on the $b$-th moment, the exact meta distribution as well as a simple yet accurate approximation based on beta distribution is provided. We derived the asymptotic gains and found that for any specific tier, the $b$-th moment as well as the variance of the conditional success probability is approximately a horizontal shifted version of that in a single-tier PPP, and hence horizontal shifted versions of each other. 

About the effect of the offloading biases, we proved that $M_1$ of the whole network is always harmed by any biasing; for multi-tier (more than $3$) HIP-based HCNs, users of certain tiers will benefit while the others suffer, which depends on the relative ratios of the biases between different tiers. The effect on the per-tier success probability can be quantified using a horizontal shift of the SIR distribution. 

The $b$-th moment of the conditional success probability under the independent ALOHA-like random base station activities was also addressed. The region of the activity probabilities in which the mean local delay of each tier remains finite is characterized by a lower bound, which was shown to be accurate enough compared to the exact one.    

Overall, the SIR meta distribution framework offers several new and interesting insights in the performance of HCNs, which helps us understand the HCNs better and hence benefits the real network design and optimization.  

\appendix
\subsection{Proof of Theorem 1}
\begin{IEEEproof}
	Let $M_{b \mid (i)}$ denote the conditional $b$-th moment of the SIR meta distribution given that the typical user at the origin connects to the $i$-th tier. Then we have
	
	\begin{equation}
	\label{eq:Mb_tot_Prob}
	M_b = \sum\limits_{i\in[K]} p_{\rm a}^{(i)}\cdot M_{b \mid (i)}.
	\end{equation}  
	
	Next, we derive the conditional $b$-th moment $M_{b \mid (i)}$.
	
	Given $R_i$ and that the typical user at the origin connects to the $i$-th tier, the conditional success probability is given by 
	\begin{equation}	
	P_{{\rm s}}^{(i)}(\theta) = \P\bigg(\frac{P_i h_o R_i^{-\alpha}}{\sum\limits_{j\neq i} \sum\limits_{{\rm x}\in\Phi_j}P_j h_{{\rm x}}R_{{\rm x}}^{-\alpha}+\sum\limits_{{\rm x} \in \Phi_i\setminus\{{\rm x}_0^{(i)}\}} P_i h_{{\rm x}} R_{{\rm x}}^{-\alpha}} >\theta \bigg),
	\end{equation}
	where $R_i$ is the distance from the typical user to the nearest BS ${\rm x}_0^{(i)}$ in the $i$-th tier, and $h_o$ is the fading coefficient associated with the link from ${\rm x}_0^{(i)}$ to the typical user. 
	
	By averaging over the fading, we get the conditional $b$-th moment of the conditional success probability, given by
	\begin{equation}
	M_{b \mid (i),R_i} = \prod_{{\rm x} \in \Phi_i} \frac{1}{(1+\theta R_i^{\alpha}R_{{\rm x}}^{-\alpha})^b} \prod_{j\neq i}\prod_{{\rm x} \in\Phi_j}\frac{1}{(1+\theta {\hat P_{ij}} R_i^{\alpha} R_{{\rm x}}^{-\alpha})^b }.
	\end{equation}
	The notation $M_{b \mid (i),R_i}$ is used to denote that the $b$-th moment is conditioned on $R_i$ and the event that the typical user connects to the $i$-th tier given $R_i$, which occurs with the probability given in \eqref{eq:condi_acc_Prob}. 
	
	By considering the conditional access probability in \eqref{eq:condi_acc_Prob}, we have the $b$-th moment of the typical user when it is served by the $i$-th tier, given by
    \allowdisplaybreaks
	\begin{align}\label{eq:Mb_typ_ith}
	M_b^{(i)} &= \E_{R_i}\bigg[\P(\iota(\gamma_1)=i \mid R_i) M_{b\mid (i),R_i}\bigg] \nonumber \\
	&= \E_{R_i}\bigg[\prod\limits_{j\neq i}e^{-\lambda_j\pi (\hat P_{ij} \hat B_{ij})^{\delta} R_i^2}  \prod_{{\rm x} \in \Phi_i} \frac{1}{(1+\theta R_i^{\alpha}R_{{\rm x}}^{-\alpha})^b} \prod_{j\neq i}\prod_{{\rm x} \in\Phi_j}\frac{1}{(1+\theta {\hat P_{ij}} R_i^{\alpha} R_{{\rm x}}^{-\alpha})^b } \bigg] \nonumber \\
	&\eqa \E_{R_i}\bigg[\prod\limits_{j\neq i}e^{-\lambda_j\pi (\hat P_{ij} \hat B_{ij})^{\delta} R_i^2} \exp\bigg(\int_{R_i}^\infty -2\lambda_i\pi\Big[1- \frac{1}{(1+\theta R_i^\alpha x_i^{-\alpha})^b} \Big]x_i\dd x_i \bigg) \nonumber\\
	&~~~~\cdot \prod\limits_{j\neq i} \exp\bigg(\int_{\hat R_j}^\infty -2\lambda_j\pi\Big[1- \frac{1}{(1+\theta \hat P_{ij} R_i^\alpha x_j^{-\alpha})^b} \Big]x_j\dd x_j \bigg)  \bigg] \nonumber \\
	&\eqb \int_0^\infty 2\lambda_i \pi r_i e^{-\lambda_i \pi r_i^2} e^{-\sum\limits_{j\neq i}\lambda_j (\hat P_{ij} \hat B_{ij})^{\delta}\pi r_i^2}  \exp\bigg(\int_{r_i}^\infty -2\lambda_i\pi\Big[1- \frac{1}{(1+\theta r_i^\alpha x_i^{-\alpha})^b} \Big]x_i\dd x_i \bigg) \nonumber\\
	&~~~~\cdot \prod\limits_{j\neq i} \exp\bigg(\int_{\hat r_j}^\infty -2\lambda_j\pi\Big[1- \frac{1}{(1+\theta \hat P_{ij} r_i^\alpha x_j^{-\alpha})^b} \Big]x_j\dd x_j \bigg) \dd r_i \nonumber \\
	&\eqc \int_0^\infty e^{-z} e^{-z\sum\limits_{j\neq i}\hat\lambda_{ij} (\hat P_{ij} \hat B_{ij})^{\delta}} \exp\bigg(-2z \int_0^1\Big(1-\frac{1}{(1+\theta u_i^\alpha)^b} \Big)u_i^{-3}\dd u_i \bigg) \nonumber \\
	&~~~~\cdot \prod\limits_{j\neq i} \exp\bigg(-2z \int_0^{(\hat P_{ij} \hat B_{ij})^{-\frac{1}{\alpha}}}\Big(1-\frac{1}{(1+\theta \hat P_{ij} u_j^\alpha)^b} \Big)u_j^{-3}\dd u_j \bigg) \dd z \nonumber \\
	&\eqd \int_0^\infty e^{-z} e^{-z\sum\limits_{j\neq i}\hat\lambda_{ij} (\hat P_{ij} \hat B_{ij})^{\delta}} \exp\bigg(-z \int_1^\infty \Big(1-\frac{1}{(1+\theta t_i^{-\alpha/2})^b} \Big)\dd t_i \bigg) \nonumber \\
	&~~~~\cdot \prod\limits_{j\neq i} \exp\bigg(-z(\hat P_{ij} \hat B_{ij})^{\delta} \int_1^\infty\Big(1-\frac{1}{(1+\theta \hat B_{ij}^{-1} t_j^{-\alpha/2})^b} \Big)\dd t_j \bigg) \dd z 
	\end{align}
	where (a) is by the PGFL of the PPP \cite[Chap.~4]{Haenggi12book}; (b) is by averaging over $R_i$; (c) is by using the variable substitution $r_i/x_i=u_i$, $r_i/x_j=u_j$ and $\lambda_i \pi r_i^2=z$, and (d) is by using the variable substitution $u_j=t_j(\hat P_{ij} \hat B_{ij})^{-\frac{1}{\alpha}}$.
	
	By using the identity
	\begin{equation}\label{eq:2F1_eqa}
	_2F_1(b,-\delta; 1-\delta; -\theta) \equiv 1+\int_1^\infty \big(1-\frac{1}{(1+\theta s^{-1/\delta})^b}\big)\dd s,
	\end{equation}
	we obtain
	\begin{equation}\label{eq:Mb-ith}
	M_b^{(i)} = \frac{1}{\sum\limits_{j}\hat\lambda_{ij} (\hat P_{ij} \hat B_{ij})^\delta~ _2F_1(b,-\delta; 1-\delta; -\theta\hat B_{ij}^{-1})}.
	\end{equation}
	
	Using 	
	\begin{equation}\label{eq:Mb-ith2}
	M_b^{(i)} = p_{\rm a}^{(i)}\cdot M_{b \mid (i)},
	\end{equation}
	and \eqref{eq:Mb_tot_Prob}, we obtain \eqref{eq:Mb-typ-Bias}.	
\end{IEEEproof}

\subsection{Proof of Theorem 2} 
\begin{IEEEproof}
	The $b$-th moment of the conditional success probability of the $i$-th tier is 
	\begin{align}\label{eq:Mb_typ_ith2}
	M_b^{(i)} &= \E_{R_i}\bigg[\P(\iota(\gamma_1)=i \mid R_i) M_{b\mid (i),R_i}\bigg] \nonumber \\
	&= \E_{R_i}\bigg[\prod\limits_{j\neq i}e^{-\lambda_j\pi (\hat P_{ij} \hat B_{ij})^{\delta} R_i^2}  \prod_{{\rm x} \in \Phi_i} \Big(\frac{p_i}{1+\theta R_i^{\alpha}R_{{\rm x}}^{-\alpha}}+1-p_i\Big)^b \prod_{j\neq i}\prod_{{\rm x} \in\Phi_j}\Big(\frac{p_j}{1+\theta {\hat P_{ij}} R_i^{\alpha} R_{{\rm x}}^{-\alpha}}+1-p_j\Big)^b  \bigg] \nonumber \\
	&\eqa \E_{R_i}\bigg[\prod\limits_{j\neq i}e^{-\lambda_j\pi (\hat P_{ij} \hat B_{ij})^{\delta} R_i^2} \exp\bigg(\int_{R_i}^\infty -2\lambda_i\pi\Big[1- \Big(1-\frac{p_i\theta R_i^\alpha x_i^{-\alpha}}{1+\theta R_i^\alpha x_i^{-\alpha}}\Big)^b \Big]x_i\dd x_i \bigg) \nonumber\\
	&~~~~\cdot \prod\limits_{j\neq i} \exp\bigg(\int_{\hat R_j}^\infty -2\lambda_j\pi\Big[1- \Big(\frac{p_j\theta\hat P_{ij}R_i^\alpha x_j^{-\alpha}}{1+\theta \hat P_{ij} R_i^\alpha x_j^{-\alpha}}\Big)^b \Big]x_j\dd x_j \bigg)  \bigg] \nonumber \\
	&\eqb \int_0^\infty 2\lambda_i \pi r_i e^{-\lambda_i \pi r_i^2} e^{-\sum\limits_{j\neq i}\lambda_j (\hat P_{ij} \hat B_{ij})^{\delta}\pi r_i^2}  \exp\bigg(\int_{r_i}^\infty -2\lambda_i\pi\Big[1- \Big(1-\frac{p_i\theta r_i^\alpha x_i^{-\alpha}}{1+\theta r_i^\alpha x_i^{-\alpha}}\Big)^b \Big]x_i\dd x_i \bigg) \nonumber\\
	&~~~~\cdot \prod\limits_{j\neq i} \exp\bigg(\int_{\hat r_j}^\infty -2\lambda_j\pi\Big[1- \Big(\frac{p_j\theta\hat P_{ij}r_i^\alpha x_j^{-\alpha}}{1+\theta \hat P_{ij} r_i^\alpha x_j^{-\alpha}}\Big)^b \Big]x_j\dd x_j \bigg) \dd r_i \nonumber \\
	&\eqc \int_0^\infty e^{-z} e^{-z\sum\limits_{j\neq i}\hat\lambda_{ij} (\hat P_{ij} \hat B_{ij})^{\delta}} \exp\bigg(-2z \int_0^1\Big(1-\Big(1-\frac{p_i \theta u_i^\alpha}{1+\theta u_i^\alpha}\Big)^b \Big)u_i^{-3}\dd u_i \bigg) \nonumber \\
	&~~~~\cdot \prod\limits_{j\neq i} \exp\bigg(-2z \int_0^{(\hat P_{ij} \hat B_{ij})^{-\frac{1}{\alpha}}} \hat\lambda_j \Big(1- \Big(1-\frac{p_j\theta\hat P_{ij}u_j^\alpha}{1+\theta \hat P_{ij} u_j^\alpha}\Big)^b \Big)u_j^{-3}\dd u_j \bigg) \dd z \nonumber \\
	&\eqd \int_0^\infty e^{-z} e^{-z\sum\limits_{j\neq i}\hat\lambda_{ij} (\hat P_{ij} \hat B_{ij})^{\delta}} \exp\bigg(-2z \int_0^1\Big(1-\Big(1-\frac{p_i \theta u_i^\alpha}{1+\theta u_i^\alpha}\Big)^b \Big)u_i^{-3}\dd u_i \bigg) \nonumber \\
	&~~~~\cdot \prod\limits_{j\neq i} \exp\bigg(-z \hat\lambda_j (\hat P_{ij} \hat B_{ij})^{\delta} 2\int_0^1\Big(1-\Big(1-\frac{p_j\theta\hat B_{ij}^{-1} u_j^\alpha}{1+\theta \hat B_{ij}^{-1} u_j^\alpha}\Big)^b \Big)u_j^{-3}\dd u_j \bigg) \dd z 
	\end{align}	
	where (a) is by the PGFL of the PPP; (b) is by averaging over $R_i$; (c) is by using the variable substitution $r_i/x_i=u_i$, $r_i/x_j=u_j$ and $\lambda_i \pi r_i^2=z$, and (d) is by using the variable substitution $u_j=u'_j(\hat P_{ij} \hat B_{ij})^{-\frac{1}{\alpha}}$.
	
	Then from \cite[Thm.~3]{MHmeta}, there is 
	\begin{equation}
	\label{eq:thm3proof}
	\int_0^1\Big(1-\Big(1-\frac{p\theta r^\alpha}{1+\theta r^\alpha}\Big)^b \Big)r^{-3}\dd r \equiv \sum_{k=1}^{\infty}\binom{b}{k} \frac{-(-p\theta)^k}{k\alpha-2} ~_2F_1(k,k-\delta; k-\delta+1; -\theta).
	\end{equation}
	
	Hence,
	\begin{equation}
	\label{eq:thm3proof2}
	M_b^{(i)} = \frac{1}{\sum\limits_j \hat\lambda_j (\hat P_{ij} \hat B_{ij})^\delta \Big(1- \sum_{k=1}^{\infty}\binom{b}{k} (-p_j\theta {\hat B_{ij}}^{-1})^k\frac{\delta}{k-\delta} ~_2F_1(k,k-\delta; k-\delta+1; -\theta {\hat B_{ij}}^{-1})\Big)}.
	\end{equation}	
\end{IEEEproof}

\bibliographystyle{IEEEtr}

\end{document}